\documentclass[12pt]{article}
\usepackage{amsmath}
\usepackage{graphicx,psfrag,epsf}
\usepackage{enumerate}
\usepackage{natbib}
\usepackage{url} % not crucial - just used below for the URL 
\usepackage{amsthm}

\newcommand{\blind}{1}

\newtheorem{theorem}{Theorem}

\theoremstyle{definition}

\newtheorem{assumption}{Assumption}

\theoremstyle{remark}

\begin{document}

\def\spacingset#1{\renewcommand{\baselinestretch}%
{#1}\small\normalsize} \spacingset{1}

%%%%%%%%%%%%%%%%%%%%%%%%%%%%%%%%%%%%%%%%%%%%%%%%%%%%%%%%%%%%%%%%%%%%%%%%%%%%%%

\if1\blind
{
  \title{\bf Peer Encouragement Designs in Causal Inference with Partial Interference and Identification of Local Average Network Effects}
  \author{Hyunseung Kang\thanks{The authors gratefully acknowledge comments from Fabrizia Mealli, Betsy Ogburn, Dylan Small, Michael Sobel, and participants at the Johns Hopkins Biostatistics Causal Inference Seminar, UPenn Causal Reading Group, and 2016 Joint Statistical Meeting Session titled ``Causal Inference in a Networked World.'' The research of Hyunseung Kang was supported in part by NSF Grant DMS1502437.}\hspace{.2cm}\\
    NSF Postdoctoral Fellow, Stanford University\\
    and \\
    Guido Imbens \\
    Economics, Stanford Graduate School of Business, Stanford University}
    \date{}
    \maketitle
} \fi

\if0\blind
{
  \bigskip
  \bigskip
  \bigskip
  \begin{center}
    {\LARGE\bf Peer Encouragement Designs in Causal Inference with Partial Interference and Identification of Local Average Network Effects}
\end{center}
  \medskip
} \fi

\bigskip
\begin{abstract}
In non-network settings, encouragement designs have been widely used to analyze causal effects of a treatment, policy, or intervention on an outcome of interest when randomizing the treatment was considered impractical or when compliance to treatment cannot be perfectly enforced. Unfortunately, such questions related to treatment compliance have received less attention in network settings and the most well-studied experimental design in networks, the two-stage randomization design, requires perfect compliance with treatment. The paper proposes a new experimental design called peer encouragement design to study network treatment effects when enforcing treatment randomization is not feasible. The key idea in peer encouragement design is the idea of personalized  encouragement, which allows point-identification of familiar estimands in the encouragement design literature. The paper also defines new causal estimands, local average network effects, that can be identified under the new design and analyzes the effect of non-compliance behavior in randomized experiments on networks.
\end{abstract}

\noindent%
\vfill

\newpage
\spacingset{1.45} % DON'T change the spacing!
\section{Introduction}
\label{sec:intro}

\subsection{Motivation: Treatment Compliance in Network Settings} \label{sec:intro_motivation}
There is a growing literature on studying causal effects of a treatment, policy, or an intervention on an outcome in network/interference settings \citep{manski_identification_1993,sacerdote_peer_2001,sobel_randomized_2006,rosenbaum_interference_2007,hudgens_toward_2008,tchetgen_tchetgen_causal_2012,aronow_estimating_2013,manski_identification_2013,ugander_graph_2013,design_eckles_2014,vanderweele_interference_2014}. The vast majority of the work focus on the case where the treatment can be assigned to all individuals in the network and compliance with treatment assignment is perfectly enforced. For example, the most well studied design in network settings to study treatment effects, the two-stage randomization design of \citet{hudgens_toward_2008}, requires the treatment of interest to be randomized to each individuals in across network blocks (see Section \ref{sec:notation} for details) and assumes that the individual perfectly complies with the treatment assignment. However, in practice, especially in the social sciences, treatment can be expensive, harmful, or unethical and consequently, enforcing treatment randomization and perfect compliance is infeasible. For example, a recent work by \citet{yi_giving_2015} studied the impact of financial aid (i.e. treatment) on student performance (i.e. outcome) among students in rural China. Since the choice to receive and accept financial aid is confounded by potentially unmeasured factors, say family socioeconomic backgrounds and student's self-perception of future earning potential, the treatment is not only difficult to randomize, but also faces issues of non-compliance. Furthermore, studies have shown presence of peer effects when evaluating student performance in school settings \citep{gaviria_school-based_2001,sacerdote_peer_2001,fletcher_social_2013} and consequently, interference has to be taken into account. \citet{evans_measuring_1992}, \citet{powell_importance_2005}, \citet{lundborg_having_2006}, \citet{fletcher_social_2010}, and \citet{an_instrumental_2015} discuss other examples of non-compliance or where treatment of interest cannot be fully randomized in school settings.

More recently, massive open online courses (MOOCs) like Coursera, edX, and Udacity, which brings the classroom learning experience to online settings, bring a deluge of data about student behavior in classrooms \citep{breslow_studying_2013, reich_rebooting_2015} and have opened new avenues for studying student behavior among their peers. Some recent studies of student behavior in MOOCs include examining the peer effect among students in an online classroom forum and its impact in overall completion of the course \citep{anderson_engaging_2014,kizilcec_encouraging_2014}. However, studying the treatment effect of say, forum engagement (i.e. treatment) in class completion (i.e. outcome), in MOOCs is fraught with issues, such as enforcing student engagement in forums and compliance \citep{anderson_engaging_2014, kizilcec_encouraging_2014}. 

\subsection{Prior Work} \label{sec:priorwork}
Even though scenarios in Section \ref{sec:intro_motivation} come up frequently, unfortunately, there is a paucity of work in studying treatment effects when treatment cannot be randomized or fully enforced in network settings. In non-network literature, a popular class of experimental designs known as encouragement designs have been used to study treatment efficacy when treatment couldn't be randomized. For example, \citet{permutt_simultaneous_1989} and \citet{sexton_clinical_1984} studied the effect of smoking (i.e. treatment) on birth weight (i.e. outcome) among pregnant mothers, not by randomizing pregnant mothers to smoke, but by randomly encouraging the mothers' physicians to discourage smoking and using this random encouragement to induce randomization on the treatment. \citet{sommer_estimating_1991} and \citet{angrist_identification_1996} generalized this problem as treatment non-compliance within the framework of instrumental variables (IV). However, all these previous studies are under non-network settings. 

Notably, \citet{sobel_randomized_2006} explored the effect of non-compliance in network settings and highlighted many challenges of applying non-compliance ideas from non-network literature into network settings. In particular, Theorems 5 and 6 of \citet{sobel_randomized_2006} showed that the standard results in non-network settings for non-compliance, such as the Wald estimator \citep{wald_fitting_1940} (or the ``IV'' estimator \citep{imbens_identification_1994,hernan_instruments_2006,baiocchi_instrumental_2014}) or the two-stage least squares (TSLS) estimator \citep{wooldridge_econometrics_2010}, failed to identify local causal effects of \citet{angrist_identification_1996}, let alone causal effects even when linear modeling assumptions were used. \citet{sobel_randomized_2006} went onto show that even the intent-to-treat (ITT) estimands, popular in non-compliance settings and a key ingredient of the Wald estimator, may not estimate a causal quantity once interference is present; see Theorems 2 and 4 of \citet{sobel_randomized_2006} for details. While recognizing the challenges, unfortunately, \citet{sobel_randomized_2006} did not provide a solution to handle non-compliance in network settings and hoped ``students of causal inference and experimental design will devote attention to these important issues in future work'' \citet{sobel_randomized_2006}. Later works by \citet{hudgens_toward_2008} and \citet{vanderweele_interference_2014} also did not address the issue of non-compliance in network experiments.

The key difficulty in resolving non-compliance in networks is the exponential amount of heterogeneity that's inherent with non-compliance in networks and this can be illustrated within the framework of principal stratification \citep{frangakis_principal_2002}. Following \citet{hudgens_toward_2008}, suppose we have a randomized experiment where a binary treatment is randomly assigned to $n$ individuals in a network. Due to interference, each individual's treatment that he/she has actually taken may be a function of all the other individuals' treatment assignment in the network. Consequently, this creates an exponential number of  principal stratas, up to $4^n$, that individuals may fall into based on their treatment assignment and received preferences. Without any reasonable restrictions on these stratas and how they compare to each other, it would be difficult to identify or even interpret an average treatment effect. One could, in theory, remove the exponential treatment heterogeneity by assuming every individual's effect of treatment assignment on treatment taken is additive, linear, and constant; see \citet{holland_causal_1988} and first-stage models of popular IV models in econometrics (e.g. Section 5.1 of \citet{wooldridge_econometrics_2010}). However, as remarked in Section 2 of \citet{angrist_identification_1996}, this oversimplification leads to unnecessarily simple treatment effects. In fact, the key to making progress in non-compliance in randomized experiments on networks is (i) placing reasonable restrictions on treatment heterogeneity and (ii) making sure that these restrictions lead to familiar, identifiable and, perhaps more importantly, interpretable causal estimands.

\subsection{Our Contribution} \label{sec:intro_contribution}
This paper addresses the scenario where treatment randomization is impractical and non-compliance may be present in network settings by using a design-based approach. Specifically, we propose a new experimental design, called the peer encouragement design, is a hybrid of encouragement designs in non-network settings and two-stage randomization of \citet{hudgens_toward_2008}. The key, novel components of the design are not so much the hybridization of the two experimental designs in the literature, but more importantly (i) the notion of personalized encouragements and (ii) identification of local network causal effects, both of which play integral roles in resolving the problem of exponential heterogeneity discussed in the previous section; in fact, we show that a simple, naive application of the two previous experimental designs will generally not lead to any reasonable causal estimands.

In proposing the peer encouragement design, we hope to achieve three goals. The first is empirical where investigators can use our design to analyze causal effects of a treatment whenever treatment randomization is not possible. Second, peer encouragement designs allow analysis of non-compliance behavior in network experiments. In particular, we extend the work of \citet{angrist_identification_1996} and \citet{sobel_randomized_2006} and show positive results regarding identification of network effects, all without making modeling assumptions, when non-compliance is present. Third, much like how encouragement designs and, more broadly, instrumental variables, are ``quasi-experimental'' designs \citep{holland_causal_1988} that serve as a middle ground between randomized experiments and observational studies in non-network settings, we hope that the peer encouragement design serves as a stepping stone for analyzing observational network data, which is generally fraught with problems\citet{shalizi_homophily_2011}, although some progress has been made \citep{hong_evaluating_2006,tchetgen_tchetgen_causal_2012,sofrygin_semiparametric_2015,forastiere_estimating_2016}.

\section{Peer Encouragement Design}
\subsection{Notation, Network Structure, and Potential Outcomes} \label{sec:notation}
Following \citet{sobel_randomized_2006}, \citet{rosenbaum_interference_2007}, \citet{hudgens_toward_2008}, \citet{tchetgen_tchetgen_causal_2012}, and \citet{vanderweele_interference_2014}, we focus on the partial interference setting where the network can be partitioned into disconnected sub-networks (i.e. blocks), such as schools in a local district where each school comprise a sub-network or online courses (or schools) in MOOC settings where each class serves as its own sub-network; see \citet{rosenbaum_interference_2007} for additional examples of partial interference where network partitioning is done temporally or spatially. Beyond partial interference, we do not assume any structure about the network nor do we assume perfect knowledge of the sub-networks. In other words, similar to \citet{hudgens_toward_2008} and \citet{tchetgen_tchetgen_causal_2012}, the results in our paper are agnostic to knowing the exact graph, which can be cumbersome to identify in practice \citep{chandrasekhar_econometrics_2011,beaman_can_2015,kim_social_2015} 

Under partial interference, we assume $N$ individuals across $B$ sub-networks/blocks in a finite population. Let $n_1,...,n_B$ be the number of individuals in each of the $B$ blocks so that $N = \sum_{j=1}^B n_k$. Let $Z_{ij}$ denote the randomized assignment (or encouragement assignment) for $j$th individual in block $i$. For each block $i$, let $\mathbf{Z}_{i} = (Z_{i1},\ldots,Z_{in_i})$ be the vector of encouragement assignments to each of the $n_i$ individuals in block $i$ and let $\mathbf{Z}_{i(j)}$ denote the $n_i - 1$ subvector of $\mathbf{Z}_{i}$ with the $j$th entry removed. Let $D_{ij}$ be the treatment received for individual $j$ in block $i$, $\mathbf{D}_{i} = (D_{i1},\ldots,D_{in_i})$ be the vector of treatment received for each of the $n_i$  individual $j$ in block $i$, and $\mathbf{D}_{i(j)}$ be the $n_i - 1$ subvector of $\mathbf{D}_{i}$ with the $j$th entry removed. Let $Y_{ij}$ denote the outcome from individual $j$ in block $i$ and the vector equivalent $\mathbf{Y}_{i} = (Y_{i1},\ldots,Y_{in_i})$. We assume the encouragement $Z_{ij}$ and the treatment $D_{ij}$ are both binary where $Z_{ij} = 1$ implies that individual $j$ in block $i$ is encouraged to a treatment value $D_{ij} = 1$ and $Z_{ij} = 0$ implies that individual $j$ is encouraged to a treatment value $D_{ij} = 0$. Let the lower cases of $Z_{ij}$ and $D_{ij}$, denoted as $z_{ij}$ and $d_{ij}$ respectively, as realizations of $Z_{ij}$ and $D_{ij}$. 

Following the potential outcomes notation for causal inference in \citet{neyman_application_1923} and \citet{rubin_estimating_1974}, let $D_{ij}(\mathbf{z}_{i})$ denote the potential treatment of individual $j$ in block $i$ if encouragements were assigned as $\mathbf{z}_{i} = (z_{11},\ldots,z_{1n_i}) $ and $\mathbf{D}_{i}(\mathbf{z}_{i}) = (D_{i1}(\mathbf{z}_{i}),\ldots,D_{in_i}(\mathbf{z}_{i}))$ be the collection of $D_{ij}$s in block $i$. Similarly, let $Y_{ij}(\mathbf{d}_{i},\mathbf{z}_{i})$ denote the potential outcome of individual $j$ in block $i$ if encouragements were $\mathbf{z}_{i}$ and treatments were $\mathbf{d}_{i}$. Let $\mathcal{F}$ contain all the potential outcomes of everyone, i.e. $\mathcal{F} = \{ (Y_{ij}(\mathbf{d}_{i},\mathbf{z}_{i}), D_{ij}(\mathbf{z}_{i})) | \mathbf{z}_{i} \in \mathcal{Z}_{n_i}, \mathbf{d}_{i} \in \mathcal{Z}_{n_i}, i=1,\ldots,B,j = 1,\ldots,n_i \}$; the set $\mathcal{Z}_{n_i} = \{(z_{i1},\ldots,z_{in_i}) | z_{ij} \in \{0,1\}\}$ denotes all values of an $n_i$ dimensional binary vector so that $\mathbf{z}_{i} \in \mathcal{Z}_{n_i}$. Let $\phi$ and $\psi$ denote two different probability measures (or mechanisms) on the set $\mathcal{Z}_{n_i}$,  i.e. $P_{\phi}(\mathbf{Z}_{i} = \mathbf{z}_{i})$ and $P_{\psi}(\mathbf{Z}_{i} = \mathbf{z}_i)$. For example, $\phi$ can represent a fair coin flip where each individual in block $i$ flips the fair coin independently of other individuals in the block and is assigned encouragement, $Z_{ij} = 1$, or no encouragement, $Z_{ij} = 0$, each with 50\% probability. This mechanism can be expressed as $P_{\phi}(\mathbf{Z}_{i} = \mathbf{z}_{i}) = \prod_{j=1}^{n_i} (1/2)^{z_{ij}} (1/2)^{1-z_{ij}} = (1/2)^{n_i}$. Finally, we also denote $I(\cdot)$ to be the indicator function.

\subsection{Average Potential Outcomes} \label{sec:avg}
From the notation, we define average potential outcomes as follows. Given a probability measure $\phi$ (or $\psi$), we define the individual average potential outcome for individual $j$ in block $i$ as the average of her potential outcomes when she is assigned the encouragement value $Z_{ij} = z$, i.e.
\begin{equation} \label{eq:margavg} 
\overline{Y}_{ij}(D_{ij}(z),z,\phi) = \sum_{z_{ij} = z, \mathbf{z}_{i(j)} \in \mathcal{Z}_{n_i - 1}} Y_{ij}(\mathbf{D}_{i}(\mathbf{z}_i), \mathbf{z}_{i}) P_{\phi}(\mathbf{Z}_{i(j)} = \mathbf{z}_{i(j)})
\end{equation}
As an example, if $z = 1$ in equation \eqref{eq:margavg}, $\overline{Y}_{ij}(D_{ij}(1),1,\phi)$ is the average of individual $j$'s potential outcomes if she were encouraged ($z = 1$) and other individuals in her block were assigned encouragement values $\mathbf{z}_{i(j)}$ where $\mathbf{z}_{i(j)}$ follows the marginal distribution of $\mathbf{Z}_{i(j)}$ specified by $\phi$. Conversely, if $z = 0$, $\overline{Y}_{ij}(D_{ij}(0),0,\phi)$ is the average of individual $j$'s potential outcomes if she were not encouraged ($z = 0$) and other individuals in her block were assigned encouragement values $\mathbf{z}_{i(j)}$ under the measure $\phi$. We also define block average potential outcomes $\overline{Y}_{i}(D_{i}(z),z,\phi) =  \sum_{j=1}^{n_i} \overline{Y}_{ij}(D_{ij}(z),z,\phi) /n_i$, 
 and population average potential outcomes $\overline{Y}(D(z),z,\phi) = \sum_{i=1}^{n} \overline{Y}_{i}(D_{i}(z),z,\phi)/B$.

We can interpret $\overline{Y}_{ij}(D_{ij}(z),z,\phi)$ as an extension of the individual intent-to-treat potential outcome in the instrumental variables literature without interference. Specifically, without interference, $\overline{Y}_{ij}(\mathbf{D}_{i}(\mathbf{z}_{i},\phi),\mathbf{z}_i) = Y_{ij}(D_{ij}(z),z,\phi) = Y_{ij}(D_{ij}(z),z)$ so that the individual average potential outcome $\overline{Y}_{ij}(D_{ij}(z),z,\phi)$ is free from its dependence on $\phi$ and equals the intent-to-treat potential outcome $Y_{ij}(D_{ij}(z),z)$. The supplementary materials also discuss some aspects about the average potential outcome in equation \eqref{eq:margavg}, including some subtle differences between marginal averaging versus conditional averaging of the potential outcomes. To simplify the discussion, we eliminate the distinction between marginal and conditional averaging by assuming an independent Bernouilli-type mechanism in the experimental design (see equation \eqref{eq:randomize} in Section \ref{sec:peerED}).

\subsection{Protocol for Peer Encouragement Design} \label{sec:peerED}
Given $N$ individuals in the network across $B$ blocks, we propose a new experimental design, called the peer encouragement design, to study causal effects of a treatment in network settings where treatment receipt cannot be completely enforced.
\begin{enumerate}
\item[1.] Let $\phi$ and $\psi$ denote two different Bernouilli (i.e. ``coin-toss'') mechanisms corresponding to different probabilities on $\mathcal{Z}_{n_i}$. Specifically, for all $\mathbf{z}_i \in \mathcal{Z}_{n_i}$, each mechanism obeys
\begin{equation} \label{eq:randomize}
 P_{\phi}(\mathbf{Z}_{i} = \mathbf{z}_i | \mathcal{F}) = P_{\phi}(\mathbf{Z}_{i} = \mathbf{z}_i) = \prod_{j=1}^{n_i} P_{\phi}(Z_{ij} = z_{ij}), \quad{} 0 < P_{\phi}(Z_{ij} = 1) < 1
 \end{equation}
We assume $P_{\phi}(Z_{ij} = z_{ij}) \neq P_{\psi}(Z_{ij} = z_{ij})$.
\item[2.] Randomly assign $K$ out of $B$ blocks to mechanism $\phi$ while the $B - K$ blocks are assigned to mechanism $\psi$.
\item[3.] For each block, assign each individual to encouragement $Z_{ij} = 1$ based on the mechanism that the block was assigned to. Encouragements must be personalized whereby for all individuals and any $z_{ij}$,
\begin{equation}  \label{eq:personalized}
D_{ij}(\mathbf{z}_{i}) = D_{ij}(\mathbf{z}_{i}') = D_{ij}(z_{ij}), \quad{} \forall z_{i(j)}, z_{i(j)}' \in \mathcal{Z}_{n_i-1}
\end{equation}
and have an effect, i.e. 
\begin{equation} \label{eq:assoc_effect}
\frac{1}{n_i} \sum_{j=1}^{n_i} D_{ij}(1) - D_{ij}(0) \neq 0
\end{equation}
\end{enumerate}
As discussed in Section \ref{sec:intro}, the peer encouragement design has elements of both a traditional encouragement design in non-network settings and the two-stage randomization design of \citet{hudgens_toward_2008} in partial interference settings. The new design randomizes the encouragement to treatment, similar to an encouragement design, and the new assign randomizes different mechanisms of encouragement in a two-stage fashion, similar to \citet{hudgens_toward_2008}, across different blocks in a network. The two mechanisms for encouragement in step 1 of the design protocol can be thought of as encouragement intensities. For example, encouragement mechanism $\phi$ corresponds to 20\% chance of encouragement while $\psi$ corresponds to 80\% chance of encouragement and each block in the network randomly receives different intensities of encouragement in step 2 of the design protocol; see also \citet{crepon_labor_2013} for related discussion in labor economics where treatment intensities were varied in a two-stage randomization design scheme. Also, the supplementary material has additional details about the Bernouilli mechanisms $\phi$ and $\psi$ in network settings versus, say, mixed assignment strategies of \citet{hudgens_toward_2008}. We note that the peer encouragement design, specifically equation \eqref{eq:randomize}, automatically satisfies ignorability of the randomization mechanism to potential outcomes as well as overlap common in causal inference \citep{imbens_causal_2015,hernan_causal_2016}. In addition, akin to the literature on encouragement designs, the encouragement should be chosen to have a non-negative treatment effect as specified in equation \eqref{eq:assoc_effect} \citep{angrist_identification_1996}. 

However, the peer encouragement design deviates from the two strands of literature, the encouragement design literature and the network experiment literature, in some important ways. First, our design requires that the encouragement must be personalized to a specific individual, as specified in equation \eqref{eq:personalized}, and to the best of our knowledge, this is the first time such an assumption has been presented in both strands of the literature. Personalizing encouragement essentially amounts to having no interference at the encouragement level so that one's uptake of the treatment through encouragement is only affected by what encouragement one was assigned to; in other words, the encouragement must be private in nature. We believe the personalized encouragement assumption serves as a balance between the technical necessity to identify familiar and interpretable casual estimands and a plausible assumption to achieve in practice. In particular, technically speaking, this assumption is key to (i) reduce treatment heterogeneity that's present in partial interference with non-compliance and (ii) to present a set of familiar and interpretable causal estimands, such as complier average causal effects (see Section \ref{sec:idLocal} for details). However, the assumption is not completely restrictively in the sense we do not make any parametric modeling assumptions between the encouragement and the treatment, which is commonly done in econometrics to analyze network data; see \citet{an_instrumental_2015} for a recent example. Also, the assumption still allows for interference between the treatment and the outcome.

Also, the personalized encouragement assumption is plausible in various settings. For example, consider the online MOOC courses example in Section \ref{sec:intro_motivation} where the goal is to study student participation in online classroom forums (i.e. treatment) on class completion (i.e. outcome). Each classroom (or school) can serve as a block in a network and the encouragement to participate in the online forum can be a prod to participate in the online forum via badges, algorithmic changes to the student's online profile, online display settings, or mobile notifications; see \citet{bond_million_2012,anderson_engaging_2014} and \citet{eckles_estimating_2015} for some recent examples where different modes of personalized encouragements were used to engage users to use online services in an online randomized experiment. These encouragements can be designed so that they can only be seen by one student; no one else in the class has any idea about the encouragement assignment of other students. As a concrete example, \citet{anderson_engaging_2014} conducted a randomized experiment to study student participation in a classroom forum in a MOOC where each student's online profile was tweaked to illicit different levels of encouragement. In some experimental conditions, other students could not see the encouragement, thereby making the personalized encouragement assumption very plausible. Indeed, with a growing trend toward personalized content, we believe the personalized encouragement assumption is very much plausible in online experiments and it is an interesting direction of future research to design different levels of personalization for experimental design in order to illicit different levels of compliance.

Although online setting seems to be the most natural platform to deliver personalized encouragements, personalized encouragements are also plausible in some non-online settings. For example, revisiting the example in Section \ref{sec:intro_motivation}, \citet{yi_giving_2015} studied the impact of financial aid (i.e. treatment) on student performance (i.e. outcome) among students in rural China. The authors of the study created a randomized incentive (i.e. encouragement) in the form of a private discussion with the school principal so that the students are more like to utilize the financial aid package. Furthermore, the financial aid package was offered to each student with a non-disclosure notice urging students  ``not to discuss [the aid] with anyone beside their guardians and the school principal'' (Section 2.2 of \citet{yi_giving_2015}). This minimized others in the network knowing about the financial aid offer, making the personalized encouragement assumption plausible. Also, another study by \citet{omalley_estimating_2014} analyzed the network effect of obesity among friends by using a gene that has been known to associate with obesity as the ``encouragement.'' Specifically, in the spirit of Mendelian randomization where the instruments/encouragements are genetic in nature \citep{davey_smith_mendelian_2003, davey_smith_mendelian_2004, lawlor_mendelian_2008}, \citet{omalley_estimating_2014} utilized the random nature of genetic allele assignment at birth as an encouragement for an individual to become obese. Since an individual's genes cannot be influenced by his friends genes given the individual's genes, the encouragement is personalized and the personalized encouragement assumption is very plausible. 

However, if the encouragement and the treatment are such that an individual's exposure can be a function of his peers receiving or not receiving encouragement, the personalized encouragement assumption can be violated. For instance, in the financial aid package example above, if the students' discussions wth the school principal about the financial aid were not private and students were allowed to share details of their financial aid package with their peers, then the personalized encouragement assumption will not hold. More broadly, in non-online settings where the personalized encouragement assumption may be suspect, one can strengthen the plausibility of the assumption by incorporating non-disclosure statements (e.g. financial aid example), creating a short timespan between the offer of a new treatment, policy or social program and the actual receipt of the treatment/policy/program so that the discussion about the encouragement amongst peers is minimized, or using private forms of encouragements, say via e-mail or mobile notifications. Ultimately, the exact way to personalize and privatize encouragement is problem-specific and researchers utilizing our design should carefully plan about not only the nature, but also the delivery of the encouragement to make the personalized encouragement assumption plausible.

Finally, in addition to the empirical implications of the personalized encouragement assumption, the personalized encouragement assumption also allows us to study the impact of non-compliance in randomized experiments on a network with familiar and interpretable estimands. Specifically, suppose we treat the encouragement $Z_{ij}$ as the treatment assigned and $D_{ij}$ as the treatment actually received in a randomized experiment. Then, under the personalized encouragement assumption, the population of individuals under our design can be partitioned into four different groups, always-takers (i.e. $D_{ij}(1) = D_{ij}(0) = 1$), compliers (i.e. $D_{ij}(1) = 1, D_{ij}(0) = 0$), never-takers (i.e. $D_{ij}(1) = D_{ij}(0) = 0$), and defiers (i.e. $D_{ij}(1) = 0, D_{ij}(0) = 1$), depending on the potential treatment values $D_{ij}(z)$. These four subgroups are the same subgroups as \citet{angrist_identification_1996} which studied non-compliance behavior in non-network settings and, as we will see in Sections \ref{sec:monotonicity} and \ref{sec:localEffect}, the personalized encouragement assumption plays a critical role in properly defining and identifying local average network effects, such as the local direct effect and the local peer effect.

\section{Causal Assumptions} \label{sec:EE}
\subsection{Network Intent-to-Treat Effects} \label{sec:ITT}
In non-network settings, it is customary in encouragement designs, or more broadly in instrumental variables analysis, to first define and identify intent-to-treat (ITT) effects, which are causal effects of the encouragement on the outcome. Similarly, in network settings, we can proceed similarly and define ITT effects in the spirt of \citet{hudgens_toward_2008}, \citet{tchetgen_tchetgen_causal_2012} and \citet{vanderweele_interference_2014}. We focus our attention on the direct intent-to-treat effect, abbreviated as $DITT$, and the peer intent-to-treat effect (also known as spillover or indirect effects), abbreviated as $PITT$. In practice, especially in the social sciences, the direct and peer effects are often quantities of great interest \citep{evans_measuring_1992,gaviria_school-based_2001, zimmerman_peer_2003, powell_importance_2005, lundborg_having_2006} and identification of direct and peer effects usually identify the total and overall ITT effects; see \citet{hudgens_toward_2008}, \citet{vanderweele_effect_2011} and the supplementary materials for definitions of total and overall ITT effects along with some well-known effect decompositions.

Formally, for any two values of the encouragement $z', z \in \{0,1\}$ where $z' \neq z$ and two mechanisms $\phi$ and $\psi$, $\phi \neq \psi$, $DITT$ and $PITT$ are defined as follows.
\begin{align*}
DITT_{i}(z',z,\phi) &= \overline{Y}_{i}(D_{i}(z'),z',\phi) - \overline{Y}_{i}(D_{i}(z),z,\phi) \\
DITT(z',z,\phi) &= \frac{1}{B} \sum_{i=1}^{B} DITT_{i}(z',z,\phi) = \overline{Y}(D(z'),z',\phi) - \overline{Y}(D(z),z,\phi) \\
PITT_{i}(z,\phi,\psi) &= \overline{Y}_{i}(D_{i}(z),z,\phi) - \overline{Y}_{i}(D_{i}(z),z,\psi) \\
PITT(z,\phi,\psi) &= \frac{1}{B} \sum_{i=1}^{B} PITT_{i}(z,\phi,\psi) = \overline{Y}(D(z),z,\phi) - \overline{Y}(D(z),z,\psi)
\end{align*}
If $z' = 1$ and $z = 0$, $DITT_{i}(1,0,\phi)$ is the direct effect of being encouraged versus not being encouraged on the outcome and $PITT_{i}(1,\phi,\psi)$ is the peer effect of being encouraged on the outcome, all within block $i$ and under measures $\phi$ and $\psi$. Similarly, $DITT(1,0,\phi)$ and $PITT(1,\phi,\psi)$ represent the direct population average ITT effect and peer population average ITT effect, respectively. One important point to mention about $PITT$s is that $PITT(0,\phi,\psi)$ and $PITT(1,\phi,\psi)$ may not equal to each other. In particular, $PITT(0,\phi,\psi)$ represents the population average peer effect of encouragement on the outcome if individuals were not encouraged and $PITT(1,\phi,\psi)$ represents the population average peer effect of encouragement on the outcome if individuals were encouraged. It is possible that the the encouraged individuals may have a stronger peer effect on the outcome compared to unencouraged individuals, or vice versa, and this distinction will become important in the presence of non-compliance (see Section \ref{sec:idLocal}). Also, note that without interference, $DITT_{i}(1,0,\phi)$ is the $i$th block average ITT effect of the encouragement, $DITT(1,0,\phi)$ would be the usual population average ITT effect and $PITT_{i}(z,\phi,\psi) = PITT(z,\phi,\psi) = 0$ for any $z = 1$ or $0$. In short, $DITT$s and $PITT$s are generalizations of the usual ITT effects common in instrumental variables to network settings.

The identification and estimation of network ITT effects like $DITT$ and $PITT$ are straightforward and directly follows from the results in \citet{hudgens_toward_2008}. We briefly re-iterate these results to aid the discussion of identifying local network effects in Section \ref{sec:idLocal}. Let $S_i$ be a binary variable that denotes which encouragement mechanism was applied to block $i$. Without loss of generality, $S_i = 1$ indicates that block $i$ received encouragement mechanism $\phi$ and $S_i = 0$ indicates that block $i$ received mechanism $\psi$. We define the following estimators for the block average and population average potential outcomes defined in Section \ref{sec:avg} for mechanism $\phi$.
\[
\widehat{\overline{Y}}_i(D_i(z),z,\phi) =  \frac{\sum_{j=1}^{n_i} Y_{ij} I(Z_{ij} = z)}{n_i P_{\phi}(Z_{ij} = z)}, \quad{} \widehat{\overline{Y}}(D(z),z,\phi) = \frac{ \sum_{i=1}^{B} \widehat{\overline{Y}}_{i}(D_i(z),z,\phi) I(S_i  = 1) }{\sum_{i=1}^B I(S_i = 1)} 
 \]
 Then, without additional assumptions beyond the assumptions which are inherent in the two-stage randomization design of \citet{hudgens_toward_2008} and are also satisfied by our peer encouragement design, we can identify $DITT$ and $PITT$
 \begin{align}
DITT(z',z,\phi) &= E\left(\widehat{\overline{Y}}(D(z'),z',\phi) - \widehat{\overline{Y}}(D(z),z,\phi) \right), \quad{} z, z' \in \{0,1\}, z \neq z'  \label{eq:idDITT} \\
PITT(z,\phi,\psi) &= E\left(\widehat{\overline{Y}}(D(z),z,\phi) - \widehat{\overline{Y}}(D(z),z,\psi) \right), \quad{} z \in \{0,1\} \label{eq:idPITT}
\end{align}
To estimate the direct ITT effect of being encouraged for a mechanism $\phi$ (i.e. $DITT(1,0,\phi)$), equation \eqref{eq:idDITT} states that the contrast between the sample averages $\widehat{{\overline{Y}}}(D(1),1,\phi)$ and $\widehat{\overline{Y}}(D(0),0,\phi)$ collected from our peer encouragement design is an unbiased estimate of $DITT(1,0,\phi)$. Also, to estimate the peer ITT effect between the two mechanisms $\phi$ and $\psi$ (i.e. $PITT(1,\phi,\psi)$), equation \eqref{eq:idPITT} states that the contrast between the sample averages $\widehat{\overline{Y}}(D(1),1,\phi)$ and $\widehat{\overline{Y}}(D(1),1,\psi)$ collected from our design is an unbiased estimate of $PITT(1,\phi,\psi)$. Note that other effects such as total and overall ITT effects can also be identified by directly applying the results from \citet{hudgens_toward_2008}.

\subsection{Network Exclusion Restriction}
To identify and estimate actual treatment effects in networks, we need to make the following identifying assumptions common in the literature on encouragement designs and instrumental variables. The first assumption, which we call network exclusion restriction, is an extension of the exclusion restriction in \citet{angrist_identification_1996} for network settings. It was also stated as Assumption 1 in \citet{sobel_randomized_2006}.
\begin{assumption}[Network Exclusion Restriction] \label{as:exclusion} For each block $i$ and for any $\mathbf{d}_{i}$, we have
\[
Y_{ij}(\mathbf{d}_i,\mathbf{z}_i) = Y_{ij}(\mathbf{d}_i,\mathbf{z}_{i}') \equiv Y_{ij}(\mathbf{d}_{i}), \quad{} \forall \mathbf{z}_{i}, \mathbf{z}_{i}' \in \mathcal{Z}_{n_i}
\]
\end{assumption}
Assumption \ref{as:exclusion} states that the outcome of individual $j$ in block $i$ does not depend on his encouragement assignment $z_{ij}$ or his peers' encouragement assignment $\mathbf{z}_{i(j)}$ so long as the treatment $d_{ij}$ as well as the treatment of others in the block $\mathbf{d}_{i(j)}$ are fixed. In short, the individual's outcome $Y_{ij}$ only depends on his and his peer's treatment, $\mathbf{d}_{i}$. Note that without interference, Assumption \ref{as:exclusion} reduces to the usual exclusion restriction in \citet{angrist_identification_1996}. 

The exclusion restriction assumption is arguably the most problematic assumption in encouragement designs or, more broadly, in instrumental variables literature, because it is unverifiable with data and often requires subject-matter expertise to rule out various causal pathways. In a similar vein, the peer encouragement design carries the same limitation. However, recent work in online network experiments allow one to design encouragements that, by design, satisfy the exclusion restriction \citep{eckles_estimating_2015}. In particular, following \citet{eckles_estimating_2015} who studied Facebook user behaviors, if the encouragement is a prod for a user to interact online in a specific way, say by writing comments, likes, or providing positive/negative on a Facebook post, the treatment is the number of comments or positive/negative feedback actually written on Facebook by the user, and the outcome is some measure of user behavior, then the treatment value of others (i.e. other users' actual comments or likes on Facebook) will only be visible if these users acted upon their encouragements (i.e. the prod to write a comment). Hence, the outcome of a particular user $Y_{ij}$ will only be a function of his encouragement, $z_{ij}$, his treatment value $d_{ij}$, and his peers' treatment vector $\mathbf{d}_{i(j)}$. More importantly, the user's outcome $Y_{ij}$ will not depend on his peers' encouragements $\mathbf{z}_{i(j)}$ and consequently, the network exclusion restriction will be as plausible as the usual exclusion restriction in non-network settings under this particular encouragement since we only have to worry about the individual's encouragement $z_{ij}$ and its role in the potential outcome $Y_{ij}(\mathbf{z}_{i}, \mathbf{d}_{i})$. Finally, there is some recent work relaxing the exclusion restriction in non-network settings \citep{kang_instrumental_2016} using multiple encouragements and it would be interesting topic of future research to see whether these ideas can be applied to network settings.

\subsection{Monotonicity} \label{sec:monotonicity}
Similar to the traditional IV literature, Assumption \ref{as:exclusion} is not sufficient to point-identify the average treatment effect \citep{imbens_identification_1994,hernan_instruments_2006, baiocchi_instrumental_2014}. As such, following \citet{angrist_identification_1996}, we make the monotonicity assumption 
\begin{assumption}[Monotonicity] \label{as:noDefiers} We assume that $D_{ij}(1) - D_{ij}(0) \geq 0$ for all $i,j$
\end{assumption}
As discussed in Section \ref{sec:peerED}, under the personalized encouragement assumption, we can characterize  local effects in terms of four different groups that individuals fall under, always-takers, compliers, never-takers, and defiers. The monotonicity assumption removes the defiers who systematically defy the encouragement that they were assigned to.

Another assumption that is commonly invoked in traditional encouragement designs and is stronger than monotonicity is the notion of one-sided compliance.
\begin{assumption}[One-Sided Compliance] \label{as:onesided} We assume $D_{ij}(z) =0$ for all $i,j$ and $z \in \{0,1\}$.
\end{assumption}
One-sided compliance states that if the individual was not encouraged, there is no way that the individual can take the treatment. This type of assumption is common in program-evaluation literature and has also been used under network settings \citep{sobel_randomized_2006}. Also, one-sided compliance implies monotonicity holds and, with respect to principal stratification, that always-takers are also not possible. As we will see in Section \ref{sec:idLocal}, in network settings with non-compliance, monotonicity and one-sided compliance each lead to identification of slightly different peer effects, which is a departure from traditional encouragement designs where both monotonicity and one-sided compliance identifies the same estimand, the complier average treatment effect \citep{angrist_identification_1996}.

\section{Identification of Local Average Network Effects}
In the next three sections, we delve in the heart of the paper, which is the identification of causal estimands in network experiments where non-compliance may be present. Specifically, utilizing the peer encouragement design, which can be used in practice to deal with settings where perfect treatment compliance is infeasible or which can serve as a vehicle to understand non-compliance in network experiments, we define local network causal effects and identification of these effects. Section \ref{sec:localAvg} discusses local averaging of potential outcomes, similar to Section \ref{sec:avg}, but local to different compliance classes discussed in Sections \ref{sec:avg} and \ref{sec:monotonicity}. Section \ref{sec:localEffect} defines local network causal effects, such as local direct effects (LDTs) and local peer effects (LPTs). Finally, Section \ref{sec:idLocal} discusses identification of the local network causal effects.

\subsection{Local Average Potential Outcomes} \label{sec:localAvg}
We start the discussion of identification of local network causal estimands by defining the average of potential outcomes of individual $j$ in block $i$ if he had treatment $d$ and rest of his peers in the block had ``natural'' treatment values $D_{i(j)}(z_{i(j)})$
\begin{equation} \label{eq:margavgLocal}
\overline{Y}_{ij}(d,D_{i(j)},\phi) = \sum_{\mathbf{z}_{i(j)} \in \mathcal{Z}_{n_i -1}} Y_{ij}(D_{ij} = d,\mathbf{D}_{i(j)}(\mathbf{z}_{i(j)})) P_{\phi}(\mathbf{Z}_{i(j)} = \mathbf{z}_{i(j)})
\end{equation}
with a slight abuse of notation where we use $D_{ij} = d$ to indicate the assignment of individual $j$'s treatment to value $d$. Note that by the personalized encouragement assumption in equation \eqref{eq:personalized}, the averaging in \eqref{eq:margavgLocal} is over the treatment $\mathbf{D}_{i(j)}$ where
\[
\mathbf{D}_{i(j)}(\mathbf{z}_{i(j)}) = (D_{i1}(Z_{i1}),\ldots,D_{ij-1}(Z_{ij-1}),D_{ij+1}(Z_{ij+1}),\ldots,D_{in_i}(Z_{in_i}))
\]
Consequently, the average potential outcome, $\overline{Y}_{ij}(d,D_{i(j)},\phi)$ in equation \eqref{eq:margavgLocal}, differs from the average potential outcome, $\overline{Y}_{ij}(D_{ij}(z),z,\phi)$, in equation \eqref{eq:margavg} because, by the network exclusion restriction in Assumption \ref{as:exclusion}, $\overline{Y}_{ij}(D_{ij}(z),z,\phi) = \overline{Y}_{ij}(D_{ij}(z),\phi)$ so that $\overline{Y}_{ij}(D_{ij}(z),\phi)$ is the average potential outcomes of individual $j$ if he and his peers took on his ``natural'' treatment value $\mathbf{D}_{i}(\mathbf{Z}_{i})$. In contrast, $\overline{Y}_{ij}(d,D_{i(j)},\phi)$ is the average potential outcomes of individual $j$ if he took a particular treatment value $d$ while his peers took on natural treatment values $\mathbf{D}_{i(j)}(\mathbf{Z}_{i(j)})$. Finally, both averages differ from average over the actual treatment value $d$, say $Y_{ij}(d_{ij},\mathbf{d}_{i(j)},\phi)$ where $\phi$ is the measure on $\mathbf{d}_{i(j)}$. Note that without interference, $\overline{Y}_{ij}(d,D_{i(j)},\phi)$ becomes the usual potential outcome $Y_{ij}(d)$ in encouragement designs. The supplementary materials discusses additional, subtle differences in averaging, which is unique to interference settings and has been discussed in other contexts, most notably by \citet{vanderweele_effect_2011}.

Given the definition of average outcome in equation \eqref{eq:margavgLocal}, we can define the usual block average potential outcome
\begin{equation} \label{eq:avgSpecial}
\overline{Y}_{i}(d,D_{i(j)},\phi) = \frac{\sum_{j=1}^{n_i} \overline{Y}_{ij}(d,D_{i(j)},\phi)}{n_i} 
\end{equation}
and the population average potential outcome, i.e. $\overline{Y}(d,D_{i(j)},\phi) = \sum_{i=1}^{B} \overline{Y}_{i}(d,D_{i(j)},\phi) / B$. In addition, we can define local block average potential outcome where the localization is specific to the four stratas discussed in Sections \ref{sec:peerED} and \ref{sec:monotonicity}. For example, we can define complier block average potential outcome as the average of potential outcomes in equation \eqref{eq:margavgLocal} among those individuals who are compliers, 
\begin{equation} \label{eq:colocalAvg}
\overline{Y}_{i}(d,D_{i(j)},\phi,Co) = \frac{\sum_{j=1}^{n_i} \overline{Y}_{ij}(d,D_{i(j)},\phi) I(D_{ij}(1) = 1, D_{ij}(0) = 0)}{\sum_{j=1}^{n_i} I(D_{ij}(1) = 1, D_{ij}(0) = 0)}
\end{equation}
We can also define the population complier average potential outcome as $\overline{Y}(d,D_{i(j)},\phi,Co) = \sum_{i=1}^{B} \overline{Y}_{i}(d,D_{i(j)},\phi,Co) / B$. Without interference, equation \eqref{eq:colocalAvg} simplifies to the block complier average potential outcome of individuals with treatment value $d$, i.e. $\overline{Y}_{i}(d,D_{i(j)},\phi,Co) = \sum_{j=1}^{n_i} Y_{ij}(d)  I(D_{ij}(1) = 1, D_{ij}(0) = 0) / \sum_{j=1}^{n_i} I(D_{ij}(1) = 1, D_{ij}(0) = 0)$. From this perspective, equation \eqref{eq:colocalAvg} can be seen as a generalization of complier average potential outcomes when interference is present. Finally, we note that we can equivalently define a similar quantity like \eqref{eq:colocalAvg} for always-takers, never-takers, and defiers. 

\subsection{Local Network Effects} \label{sec:localEffect}
Once we defined the local average potential outcomes in \eqref{eq:colocalAvg}, we can define local network effects such as local direct treatment effects, denoted as $LDT$s, and local peer treatment effects, denoted as $LPT$s. To the best of our knowledge, this is the first definition of a local effect in network randomized experiments. Consider any two values of the treatment $d',d \in \{0,1\}$ where $d' \neq d$ and the two mechanisms $\phi$ and $\psi$ where $\phi \neq \psi$. Then, we define the complier direct treatment effects and complier peer treatment effects as follows.
\begin{align*}
LDT_{i}(d',d,\phi,Co) &= \overline{Y}_{i}(d',D_{i(j)},\phi,Co) - \overline{Y}_{i}(d,D_{i(j)},\phi,Co) \\
LDT(d',d,\phi,Co) &= \frac{1}{B} \sum_{i=1}^{B} LDT_{i}(d',D_{i(j)},\phi,Co) = \overline{Y}(d',D_{i(j)},\phi,Co) - \overline{Y}(d,D_{i(j)},\phi,Co)\\
LPT_{i}(d,\phi,\psi,Co) &= \overline{Y}_{i}(d,D_{i(j)},\phi,Co) - \overline{Y}_{i}(d,D_{i(j)},\psi,Co) \\
LPT(d,\phi,\psi,Co) &= \frac{1}{B} \sum_{i=1}^B LPT_{i}(d,\phi,\psi,Co) = \overline{Y}(d,D_{i(j)},\phi,Co) - \overline{Y}(d,D_{i(j)},\psi,Co)
\end{align*}
The quantity $LDT(1,0,\phi,Co)$ is the population average direct causal effect among compliers who take treatment $d = 1$ over $d = 0$ while their peers take on natural treatment values $D_{i(j)}(z_{i(j)})$. The quantity $LPT(1,\phi,\psi,Co)$ is the population average peer causal effect among compliers who take treatment $d = 1$ while their peers take on natural treatment values $D_{i(j)}$. If a treatment is supposed to confer benefits in the form of a high outcome value, a high value of $LDT(1,0,\phi,Co)$ would indicate that the treatment has a strong direct effect among individuals who comply with the encouragement. Also, a high value of $LPT(1,\phi,\psi,Co)$ would indicate that the peer's treatments have strong influences for the outcomes of compliers. Again, similar to equation \eqref{eq:colocalAvg}, we can also define $LDT$s and $LPT$s for always-takers, never-takers, and defiers.

The local effects like $LDT$s and $LPT$s differ from intent-to-treat effects like $DITT$s and $PITT$s and other network estimands that have been defined in the literature. In particular, $LDT$s and $LPT$s describe the efficacy of an individual's treatment while $DITT$s and $PITT$s describe the efficacy of the encouragement. Also, $LDT$s and $LPT$s differ from the direct effects and indirect/spillover effects of \citet{hudgens_toward_2008} in that $LDT$s and $LPT$s only average over the potential outcomes over subpopulations of individuals, say compliers, and over the peers' natural treatment values. The $LDT$s and $LPT$s 
will equal the direct and spillover effects of \citet{hudgens_toward_2008} if we average across all the individuals and if $D_{ij}(z_{ij}) = z_{ij}$ for every $z_{ij}$ and $i,j$, i.e. if everyone is a complier; otherwise, some potential outcomes may not be observed and consequently, these potential outcomes may not part of the local average outcome in equation \eqref{eq:margavgLocal}. Ideally, it would be attractive to estimate the direct treatment effects and indirect/spillover treatment effects of \citet{hudgens_toward_2008}. However, if randomization of the treatment is infeasible, or treatment compliance cannot be fully enforced, it would be difficult, if not impossible, to identify those effects and we are left with what other quantity can one identify, in our case $LDT$s and $LPT$s. These $LDT$s and $LPTS$ may be the second best estimates to the direct and spillover treatment effects, much like how the local average treatment effect (LATE) in non-network settings \citep{imbens_identification_1994,angrist_identification_1996} is the second best estimate to the average treatment effect (ATE) whenever non-compliance is present\citep{imbens_better_2010,imbens_instrumental_2014,baiocchi_instrumental_2014}.

We also define the peer treatment effect for everyone based on the average potential outcome in equation \eqref{eq:avgSpecial}. This estimand will be useful when we describe identification under one-sided compliance in Section \ref{sec:idLocal}.
\begin{align*}
LPT_{i}(d,\phi,\psi) &= \overline{Y}_{i}(d,D_{i(j)},\phi) - \overline{Y}_{i}(d,D_{i(j)},\psi) \\
LPT(d,\phi,\psi) &= \frac{1}{B} \sum_{i=1}^B LPT_{i}(d,\phi,\psi) = \overline{Y}(d,D_{i(j)},\phi) - \overline{Y}(d,D_{i(j)},\psi)
\end{align*}
The difference between $LPT_i(d,\phi,\psi)$ and $LPT_i(d,\phi,\psi,Co)$ is that the average potential outcome in $LPT_i(d,\phi,\psi)$ is across everyone in the block while $LPT_i(d,\phi,\psi,Co)$ is only for compliers. Furthermore, the quantity $LPT_i(d,\phi,\psi)$ describes the causal effect of an identifiable group of people in the population, which is everyone, while the quantity $LPT_i(d,\phi,\psi,Co)$ describes the causal effect of an unidentifiable group of the population, the population of compliers. Hence, $LPT_i(d,\phi,\psi)$ avoids some of the concerns over local causal estimands about identifying an unidentifiable subset of the population \citep{hernan_instruments_2006,deaton_instruments_2010,swanson_think_2014}.

\subsection{Identification of $LDT$s and $LPT$s} \label{sec:idLocal}
Given the peer encouragement design and the assumptions regarding exclusion restriction and monotonicity, we can start to make progress on identifying local network effects, such as $LDT$s and $LPTs$. To begin, we first define the average causal effect of the encouragement on the treatment, denoted as $ET$s,
\begin{align*}
ET_i(z',z) &= \frac{1}{n_i} \sum_{j=1}^{n_i} D_{ij}(z') - D_{ij}(z) \\
ET(z',z) &= \frac{1}{B} \sum_{i=1}^{B} ET_i(z',z) = \frac{1}{B} \sum_{i=1}^{B}  \frac{1}{n_i}\sum_{j=n_i} D_{ij}(z') - D_{ij}(z)
\end{align*}
$ET_i(z',z)$ is the block $i$'s average causal effect of the encouragement on the treatment while $ET(z',z)$ is the population average causal effect of the encouragement on the treatment. Both quantities can be identified and estimated by taking sample averages of the treatment under different values of the encouragement, i.e.
\[
ET_i(z',z) = E\left( \frac{ \sum_{j=1}^{n_i} D_{ij} I( Z_{ij} = z') }{\sum_{j=1}^{n_i} I(Z_{ij} = z') }- \frac{ \sum_{j=1}^{n_i} D_{ij} I(Z_{ij} = z) }{\sum_{j=1}^{n_i} I(Z_{ij} = z)} \right)
\]
Note that because of the personalized encouragement assumption in equation \eqref{eq:personalized}, the $ET$s do not depend on a measure, say $\phi$ or $\psi$. 

With the intent-to-treat effects defined in Section \ref{sec:ITT}, we can identify the local treatment effects defined in Section\ref{sec:localEffect}. First, Theorem \ref{thm:1} shows that the ratio of $DITT$ over $ET$ identifies the local direct treatment effect among compliers. 
\begin{theorem}[Identification of Complier $LDT$] \label{thm:1} Suppose we use the peer encouragement design and Assumptions \ref{as:exclusion} and \ref{as:noDefiers} hold. Then, for any measure $\phi$
\begin{equation} \label{eq:idLDT}
\frac{DITT(1,0,\phi)}{ET(1,0)} = LDT(1,0,\phi,Co)
\end{equation}
\end{theorem}
Theorem \ref{thm:1} is a generalization of the classic result about ratio estimators in instrumental variables, i.e. the Wald estimator. It states that under interference, the ratio of the intent-to-treat effects identifies the local direct treatment effect among compliers; without interference, Theorem \ref{thm:1} reduces to Proposition 1 in \citet{angrist_identification_1996} where the ratio of the intent-to-treat effects identifies the local treatment effect among compliers. The interference forces us to consider the individual average potential outcomes as defined in Section \ref{sec:localAvg} and pool their individual effects according to some measure $\phi$. Note that we can equivalently identify $LDT(1,0,\psi,Co)$ under a different measure $\psi$ by using $DITT(1,0,\psi)$ instead of $DITT(1,0,\phi)$ in \eqref{eq:idLDT}. Finally, we note that in \citet{sobel_randomized_2006}, the author identified a different and arguably more complex estimand, using the Wald estimator, primarily because personalized encouragement assumption was absent in his work and therefore, his causal estimand was much more difficult to interpret.

Second, Theorem \ref{thm:2} shows that the ratio of difference in $PITT$s over $ET$ identifies the difference in local peer treatment effect among compliers. 
\begin{theorem}[Identification of Difference in Complier $LPT$] \label{thm:2} Suppose we use the peer encouragement design and Assumptions \ref{as:exclusion} and \ref{as:noDefiers} hold. Then, 
\begin{equation} \label{eq:idLPT}
\frac{PITT(1,\psi,\phi) - PITT(0,\psi,\phi)}{ET(1,0)} = LPT(1,\psi,\phi,Co) - LPT(0,\psi,\phi,Co)
\end{equation}
\end{theorem}
Unlike Theorem \ref{thm:1}, which is a natural generalization of the results in \citet{angrist_identification_1996}, the result in Theorem \ref{thm:2} is specific to the case when interference is present. Specifically, interference presents a new set of estimand, the peer effect, and Theorem \ref{thm:2} shows that the difference of local peer treatment effects can be identified by using the difference of intent-to-treatment effects, scaled by $ET(1,0)$. While initially, the difference between $LPT$s may not be useful, in most practical applications, it will not be the case that $LPT(1,\psi,\phi,Co)$ and $LPT(0,\psi,\phi,Co)$ would be identically in magnitude and opposite in sign. Consequently, if the ratio of $PITT(1,\psi,\phi) - PITT(0,\psi,\phi)$ over $ET(1,0)$ is not zero, there is reason to believe that there is some peer effect of the treatment in the presence of non-compliance. Furthermore, the sign of the local peer effect can indicate a difference in magnitude between the peer effect when one is treated, $LPT(1,\psi,\phi,Co)$, and when one isn't treated, $LPT(0,\psi,\phi,Co)$.

The main reason that we cannot identify each component of the $LPT$s under the peer encouragement design with network exclusion restriction and monotonicity is because peer treatment effects require the individual treatment value to be held fixed. This would imply that the encouragement should have no effect on the treatment, defeating the original purpose of the encouragement which attempts to provide a random nudge for individuals to take treatment. In addition, the peer effect among compliers would not be identifiable since compliers are those that change treatment assignment according to their encouragement assignment. Indeed, there is some sense that monotonicity is not strong enough to tease out specific peer effects.

In light of these discussions about the identifiability of $LPT$s, we show in Theorem \ref{thm:3} that under a more strict version of monotonicity where we assume one-sided compliance, we can identify the local peer treatment effect $LPT(d,\phi,\psi)$. 
\begin{theorem} \label{thm:3} Suppose we use the peer encouragement design and Assumptions \ref{as:exclusion} and \ref{as:onesided} hold. Then, we have 
\begin{equation} \label{eq:idonesidedLPT}
PITT(0,\phi,\psi) = LPT(0,\phi,\psi)
\end{equation}
\end{theorem}
Theorem \ref{thm:3} states that under one-sided compliance, the peer intent-to-treat effect is equal to the local peer treatment effect for everyone. This is because those who are assigned an encouragement value of $Z_{ij} = 0$ can never receive the treatment and thus, the treatment value is fixed at $D_{ij} = 0$. This allows us to identify the peer treatment effect, which requires the treatment value to be fixed. We remark that if the one-sided compliance is designed such that $D_{ij}(1) = 1$, then by using the same argument as Theorem \ref{thm:3}, we obtain $PITT(1,\phi,\psi) = LPT(1,\phi,\psi)$. Also, as mentioned before, unlike the usual encouragement design, assuming monotonicity or one-sided compliance leads to different identification.

The identification results from Theorems \ref{thm:1} to \ref{thm:3} highlight many ways one can obtain some evidence of direct or peer treatment effects from ITT effects. While the peer encouragement design is not a perfect solution in practice in the sense that we can not recover the original direct or peer treatment effect and instead, we have to settle with local versions of said effects, the results provide familiar and interpretive estimands of the treatment effect in settings where the treatment cannot be randomly assigned. From the perspective of non-compliance, the analysis of the peer encouragement design shows that non-compliance in network settings can introduce new complexities, such as the identification of local difference in peer effects and different identifications under different versions of monotonicity, while also generalizing familiar estimand in the IV literature, such as the generalization of the complier average treatment effect to the direct complier average treatment effect.

We briefly remark on estimation and inference of the $LDT$s and $LPT$s. For point-estimation, a simple strategy for obtaining estimates of $LDT$s and $LPTs$ is to use the plug-in approach where we replace population estimands of the intent-to-treat estimands in equations \eqref{eq:idLDT}-\eqref{eq:idonesidedLPT} with the sample versions outlined in Section \ref{sec:ITT}. Unfortunately, for inference like confidence intervals, interference makes this problem difficult and one may have to (i) make stratified interference assumption \citep{hudgens_toward_2008} about the intent-to-treat effects, (ii) use bounds if the outcome is binary \citep{tchetgen_tchetgen_causal_2012}, or (iii) resort to asymptotic approximations \citep{liu_large_2014} to make progress. However, once these simplifying assumptions are made, one can use the delta method to obtain a first-order approximation of confidence intervals and standard errors; see Chapter 23 of \citet{imbens_causal_2015} on using the delta method for instrumental variables with ITT estimates. Overall, the key to obtaining inference for the local causal network estimands rely on obtaining accurate inference for the ITT effects in the presence of interference and we leave it as a topic of future research.

\section{Discussion}
In this paper, we present an experimental design, the peer encouragement design, to study network effects when randomizing the treatment is infeasible and treatment compliance cannot be enforced. The new design is motivated by approaches in causal inference with partial interference as well as encouragement designs popular in instrumental variables analysis. However, the new design deviates from the prior literature by introducing the notion of personalized encouragement and local estimands in network settings, both of which work to reduce the exponential heterogeneity that is present in network settings with non-compliance. The peer encouragement design serves not only for future empirical work to study treatment efficacy whenever treatment randomization is infeasible, but also to study the effect of non-compliance when randomized experiments are conducted on networks.

\bigskip
\begin{center}
{\large\bf SUPPLEMENTARY MATERIAL}
\end{center}

\begin{description}

\item[Supplementary Materials for Peer Encouragement Design] Supplementary materials contain additional details of the experimental design. (.pdf file)

\end{description}

\section{Appendix}
\begin{proof}[Proof of Theorem \ref{thm:1}] By the exclusion restriction and personalized encouragement, we have
\begin{align*}
&\overline{Y}_{ij}(D_{ij}(1),1,\phi) - \overline{Y}_{ij}(D_{ij}(0),0,\phi) \\
=& \overline{Y}_{ij}(D_{ij}(1),\phi) - \overline{Y}_{ij}(D_{ij}(0),\phi) \\
=& \sum_{\mathbf{z}_{i(j)} \in \mathcal{Z}_{n_i - 1}} \left(Y_{ij}(D_{ij}(1) = 1,\mathbf{D}_{i(j)}(\mathbf{z}_{i(j)})) D_{ij}(1) + Y_{ij}(D_{ij}(1) = 0,\mathbf{D}_{i(j)}(\mathbf{z}_{i(j)})) (1 - D_{ij}(1)) \right) P_{\phi}(\mathbf{Z}_{i(j)} = \mathbf{z}_{i(j)}) \\
\quad{}&- \left(Y_{ij}(D_{ij}(0) = 1,\mathbf{D}_{i(j)}(\mathbf{z}_{i(j)})) D_{ij}(0) + Y_{ij}(D_{ij}(0) = 0,\mathbf{D}_{i(j)}(\mathbf{z}_{i(j)})) (1 - D_{ij}(0)) \right) P_{\phi}(\mathbf{Z}_{i(j)} = \mathbf{z}_{i(j)}) \\
=& (D_{ij}(1) - D_{ij}(0)) \sum_{\mathbf{z}_{i(j)} \in \mathcal{Z}_{n_i - 1}} \left(Y_{ij}(D_{ij}(1) = 1,\mathbf{D}_{i(j)}(\mathbf{z}_{i(j)})) - Y_{ij}(D_{ij}(0) = 0,\mathbf{D}_{i(j)}(\mathbf{z}_{i(j)}))\right)P_{\phi}(\mathbf{Z}_{i(j)} = \mathbf{z}_{i(j)}) \\
=&(D_{ij}(1) - D_{ij}(0)) (\overline{Y}_{ij}(1,D_{i(j)},\phi) - \overline{Y}_{ij}(1,D_{i(j)},\phi))
\end{align*}
Then, under monotonicity, we can take the blockwise average to obtain
\begin{align*}
&DITT_{i}(1,0,\phi) \\
=& \frac{1}{n_i} \sum_{j=1}^{n_i}\overline{Y}_{ij}(D_{ij}(1),1,\phi) - \overline{Y}_{ij}(D_{ij}(0),0,\phi)  \\
=& \frac{1}{n_i} \sum_{j=1}^{n_i} (D_{ij}(1) - D_{ij}(0)) (\overline{Y}_{ij}(1,D_{i(j)},\phi) - \overline{Y}_{ij}(0,D_{i(j)},\phi)) \\
=&\left( \frac{\sum_{j=1}^{n_i} I(D_{ij}(1) = 1, D_{ij}(0) = 0)}{n_i} \right) \sum_{j: D_{ij}(1) = 1, D_{ij}(0) = 0} \overline{Y}_{ij}(1,D_{i(j)},\phi) - \overline{Y}_{ij}(0,D_{i(j)},\phi)) \\
=& \left( \frac{\sum_{j=1}^{n_i} I(D_{ij}(1) = 1, D_{ij}(0) = 0)}{n_i} \right) LDT_i(1,0,\phi,Co)
\end{align*}
Diving the above quantity by $ET_{i}(1,0)$ and summing across the blocks will give you identification of $LDT(1,0,\phi,Co)$.
\end{proof}

\begin{proof}[Proof of Theorem \ref{thm:2}]
By the exclusion restriction and personalized encouragement, we have
\begin{align*}
&\overline{Y}_{ij}(D_{ij}(1),1,\phi) - \overline{Y}_{ij}(D_{ij}(1),1,\psi) \\
=& \overline{Y}_{ij}(D_{ij}(1),\phi) - \overline{Y}_{ij}(D_{ij}(1),\psi) \\
=& \sum_{\mathbf{z}_{i(j)} \in \mathcal{Z}_{n_i - 1}} \left(Y_{ij}(D_{ij}(1) = 1,D_{i(j)}(z_{i(j)})) D_{ij}(1) + Y_{ij}(D_{ij}(1) = 0,D_{i(j)}(z_{i(j)})) (1 - D_{ij}(1)) \right) P_{\phi}(\mathbf{Z}_{i(j)} = \mathbf{z}_{i(j)}) \\
\quad{}&- \left(Y_{ij}(D_{ij}(1) = 1,D_{i(j)}(z_{i(j)})) D_{ij}(1) + Y_{ij}(D_{ij}(0) = 0,D_{i(j)}(z_{i(j)})) (1 - D_{ij}(1)) \right) P_{\psi}(\mathbf{Z}_{i(j)} = \mathbf{z}_{i(j)}) \\
=& D_{ij}(1) ( \overline{Y}_{ij}(1,D_{i(j)},\phi) - \overline{Y}_{ij}(1,D_{i(j)},\psi)) + (1 - D_{ij}(1)) ( \overline{Y}_{ij}(0,D_{i(j)},\phi) - \overline{Y}_{ij}(0,D_{i(j)},\psi))
\end{align*}
A similar algebraic manipulation also leads to 
\begin{align*}
&\overline{Y}_{ij}(D_{ij}(0),0,\phi) - \overline{Y}_{ij}(D_{ij}(0),0,\psi) \\
=& D_{ij}(0) ( \overline{Y}_{ij}(1,D_{i(j)},\phi) - \overline{Y}_{ij}(1,D_{i(j)},\psi)) + (1 - D_{ij}(0)) ( \overline{Y}_{ij}(0,D_{i(j)},\phi) - \overline{Y}_{ij}(0,D_{i(j)},\psi))
\end{align*}
Then, under monotonicity, we can take the block average to obtain 
\begin{align*}
&PITT_{i}(1,\phi,\psi) - PITT_i(0,\phi,\psi) \\
=& \frac{1}{n_i} \sum_{j=1}^{n_i} \overline{Y}_{ij}(D_{ij}(1),\phi) - \overline{Y}_{ij}(D_{ij}(1),\psi) - \frac{1}{n_i} \sum_{j=1}^{n_i} \overline{Y}_{ij}(D_{ij}(0),\phi) - \overline{Y}_{ij}(D_{ij}(0),\psi) \\
=& \frac{1}{n_i} \sum_{j=1}^{n_i} D_{ij}(1) ( \overline{Y}_{ij}(1,D_{i(j)},\phi) - \overline{Y}_{ij}(1,D_{i(j)},\psi)) + (1 - D_{ij}(1)) ( \overline{Y}_{ij}(0,D_{i(j)},\phi) - \overline{Y}_{ij}(0,D_{i(j)},\psi)) \\
&\quad{} - \frac{1}{n_i} \sum_{j=1}^{n_i}D_{ij}(0) ( \overline{Y}_{ij}(1,D_{i(j)},\phi) - \overline{Y}_{ij}(1,D_{i(j)},\psi)) + (1 - D_{ij}(0)) ( \overline{Y}_{ij}(0,D_{i(j)},\phi) - \overline{Y}_{ij}(0,D_{i(j)},\psi)) \\
=& \frac{1}{n_i} \sum_{j=1}^{n_i}  (D_{ij}(1) - D_{ij}(0)) ( \overline{Y}_{ij}(1,D_{i(j)},\phi) - \overline{Y}_{ij}(1,D_{i(j)},\psi)) \\
\quad{}&- \frac{1}{n_i} \sum_{j=1}^{n_i}  (D_{ij}(1) - D_{ij}(0)) ( \overline{Y}_{ij}(0,D_{i(j)},\phi) - \overline{Y}_{ij}(0,D_{i(j)},\psi)) \\
=& \left(\frac{\sum_{j=1}^{n_i} I(D_{ij}(1) = 1,D_{ij}(0) = 0)}{n_i} \right) LPT_i(1,\phi,\psi,Co)  \\
\quad{}& -\left(\frac{\sum_{j=1}^{n_i} I(D_{ij}(1) = 1,D_{ij}(0) = 0)}{n_i} \right) LPT_i(0,\phi,\psi,Co)
\end{align*}
Diving the above quantity by $ET_i(1,0)$ and summing across the blocks will give you the identification of $LPT(1,\phi,\psi,Co) - LPT(0,\phi,\psi,Co)$.
\end{proof}

\begin{proof}[Proof of Theorem \ref{thm:3}]
From the proof from Theorem \ref{thm:2}, we have
\begin{align*}
&\overline{Y}_{ij}(D_{ij}(0),0,\phi) - \overline{Y}_{ij}(D_{ij}(0),0,\psi) \\
=& D_{ij}(0) ( \overline{Y}_{ij}(1,D_{i(j)},\phi) - \overline{Y}_{ij}(1,D_{i(j)},\psi)) + (1 - D_{ij}(0)) ( \overline{Y}_{ij}(0,D_{i(j)},\phi) - \overline{Y}_{ij}(0,D_{i(j)},\psi)) 
\end{align*}
By one-sided compliance, we have
\[
\overline{Y}_{ij}(D_{ij}(0),0,\phi) - \overline{Y}_{ij}(D_{ij}(0),0,\psi) = ( \overline{Y}_{ij}(0,D_{i(j)},\phi) - \overline{Y}_{ij}(0,D_{i(j)},\psi)) 
\]
Then, taking the block average gives us
\[
PITT_{i}(0,\phi,\psi) = \frac{1}{n_i} \sum_{j=1}^{n_i} ( \overline{Y}_{ij}(0,D_{i(j)},\phi) - \overline{Y}_{ij}(0,D_{i(j)},\psi)) = \frac{1}{n_i} LPT(0,\phi,\psi)
\]
\end{proof}

\bibliographystyle{chicago}
\bibliography{network}
\end{document}